\documentclass[12pt]{amsart}
\usepackage{amsmath}
  \usepackage{graphics} 
  \usepackage{epsfig} 
\usepackage{graphicx}  \usepackage{epstopdf}
 \usepackage[colorlinks=true]{hyperref}
\hypersetup{urlcolor=blue, citecolor=red}

\newcommand{\F}{\mathbb{F}}

\newcommand{\Q}{\overline{Q}}
\newcommand{\N}{\overline{N}}
\newtheorem{theorem}{Theorem}[section]

\newtheorem{lemma}[theorem]{Lemma}
\newtheorem{proposition}[theorem]{Proposition}

\theoremstyle{definition}

\begin{document}

\title{Some ternary cubic two-weight codes }
      \author[MinJia Shi, Daitao Huang, and Patrick Sol\'e]{}
 \email{smjwcl.good@163.com}
 \email{dtHuang666@163.com}
 \email{sole@enst.fr}
 \thanks{The first author is supported by NNSF of China (61672036),
Technology Foundation for Selected Overseas Chinese Scholar, Ministry of Personnel of China (05015133) and the Open Research Fund of National Mobile Communications Research Laboratory, Southeast University (2015D11) and Key projects of support program for outstanding young talents in Colleges and Universities (gxyqZD2016008). }

\thanks{$^{**}$ Corresponding author}

\maketitle

\medskip \centerline{\scshape Minjia Shi$^{*}$}
\medskip
{\footnotesize
\address\center{School of Mathematical Sciences, Anhui University,
Anhui, 230601, P. R. China, National Mobile Communications Research Laboratory, Southeast University, 210096, Nanjing, P. R. China}\\
  }
\medskip \centerline{\scshape Daitao Huang}
{\footnotesize
\centerline{School of Mathematical Sciences, Anhui University,
Anhui, 230601, P. R. China}}

\medskip
\medskip \centerline{\scshape Patrick Sol\'e}
\medskip
{\footnotesize
 \centerline{CNRS/LAGA, University Paris 8, 93 526 Saint-Denis, France}

\bigskip


\begin{abstract}
We study trace codes with defining set $L,$ a subgroup of the multiplicative group of an extension of degree $m$ of the alphabet ring
$\mathbb{F}_3+u\mathbb{F}_3+u^{2}\mathbb{F}_{3},$ with $u^{3}=1.$  These codes are abelian, and their ternary images are quasi-cyclic of co-index three (a.k.a. cubic codes).
Their Lee weight distributions are computed by using Gauss sums.
 These codes have three nonzero weights when $m$ is singly-even and $|L|=\frac{3^{3m}-3^{2m}}{2}.$
 When $m$ is odd, and $|L|=\frac{3^{3m}-3^{2m}}{2}$, or $|L|={3^{3m}-3^{2m}}$ and $m$ is a positive integer, we obtain two new infinite families of two-weight codes which are optimal. Applications of the image codes to secret sharing schemes are also given.
\end{abstract}
\maketitle

\section{Introduction}
Linear codes with few weights are of special interest in secret sharing schemes, as the associated access structures can be completely determined analytically \cite{DD,YD}.
Due to their connections with strongly regular graphs, association schemes \cite{D} and difference sets \cite{CG,CW}, two-weight codes and three-weight codes
have been studied in many articles, see for instance \cite{DGZ,DKS,DLLZ,LLHT,ZD2}. However, most constructions, so far, have used cyclic codes \cite{DLLZ,ZD2}.

A recent study introduced trace codes over rings. In a series of papers \cite{SLS2,SLS1,SWLS}, the authors have
extended the notion of trace codes from fields to rings. The alphabet ring we consider here is $R=\mathbb{F}_3+u\mathbb{F}_3+u^{2}\mathbb{F}_{3}$ with $u^{3}=1.$
This ring is instrumental in constructing quasi-cyclic codes of co-index $3$ \cite{B,LS}, called cubic codes in \cite{B}. While the construction we use here is not exactly
the cubic construction of \cite{LS}, the ternary codes we construct here are still cubic. In \cite{SWS}, the authors constructed optimal ternary codes from one-Lee-weight codes and two-Lee-weight codes over the ring $R$.

In the present paper, we construct a trace code with a defining set
$$L =\{d_1, d_2, \cdots, d_n\}\subseteq \mathcal{R}^*,$$ by the formula
$C_L=\{(Tr(ax))_{x\in L}\mid a\in \mathcal{R}\}$, where $\mathcal{R}$ denotes an extension of $R$ of degree $m$, and $Tr()$ denotes a linear function from
$\mathcal{R}$ down to $R$. This code has two or three weights, depending on the choice of the parameters $L$ and $m.$
By varying $L$ and $R$, various codes can be
constructed. Compared with the linear codes in \cite{SWS}, the two-weight codes we construct here are different, and the method is also different. The localizing set of our abelian code is not a cyclic group, but it is an abelian group. It is related to quadratic
residues in an extension of degree $m$ of $\F_3$, which makes  quadratic Gauss sums appear naturally in the weight distribution analysis.
When $m$ is odd, $L=L'$, or $L=\mathcal{R}^{*}$ and $m$ is an integer, we obtain an infinite family of two-weight codes which satisfy the optimality.
We show that, both in the three-weight and two-weight cases, the ternary image has a very nice support inclusion structure
which makes it suitable for use in a Massey secret sharing scheme  \cite{DY2,YD}. Indeed, we can show that all nonzero codewords of the ternary images are minimal
for the partial order on codewords defined by inclusion of supports.

The paper is organized as follows. Section 2 collects the basic notions and notations needed. Section 3 shows that the trace codes are abelian.
Section 4 recalls and reproves some results on Gaussian periods. Section 5 computes the weight distribution of our codes,
building on the character sum evaluation of the preceding section. Section 6 discusses the optimality of the ternary linear codes.  The minimum distance of the dual codes is discussed in Section 7. Section 8 determines the support
structure of the ternary image and describes an application to secret sharing schemes. Section 9 concludes this paper.

\section{Basic notions and notation}
\subsection{Rings}
We consider the ring $R=\F_3+u\F_3+u^{2}\F_3$ with $u^{3}=1.$ Note that, by Fermat little theorem, $u^{3}-1=(u-1)^3.$ This implies that $R$ is a local ring with
the following lattice of ideals:
$$0\subseteq \langle 1+u+u^2\rangle=\{0,1+u+u^2,2+2u+2u^2\}\subseteq \langle u-1 \rangle=\{(u-1)a:a\in R\}\subseteq R.$$
Hence, $\langle u-1 \rangle$ is the unique maximum ideal of $R.$
Given a positive integer $m$, we can construct the ring extension of $R$ of degree $m$ given by ${\mathcal{R}}=\F_{3^m}+u\F_{3^m}+u^{2}\F_{3^m}.$
There is a Frobenius operator $F$ which maps $a+ub+u^{2}c$ onto $a^{3}+ub^{3}+u^{2}c^{3},$ for all $a,b,c \in \F_{3^{m}}.$ The \emph{Trace function,} denoted by $Tr$, is defined as
$$Tr=\sum_{j=0}^{m-1}F^j.$$
It is easy to check that $$Tr(a+ub+u^{2}c)=tr(a)+utr(b)+u^{2}tr(c)$$ for $a,b,c \in \F_{3^{m}}.$
Here $tr()$ denotes the standard trace of $\F_{3^m}.$

The ring ${\mathcal{R} }$ is local with  maximal ideal $M=\langle u-1 \rangle$ and ${\mathcal{R}}/M\cong\F_{3^m}$.
The group of units ${\mathcal{R}}^*=\F_{3^{m}}^{*}\times\F_{3^{m}}\times\F_{3^{m}},$  as a multiplicative group, is isomorphic to the product of a cyclic group of
order $3^m-1$ by two elementary abelian groups
of order $3^m.$
Denoting by ${\mathcal{Q}},$ and ${\mathcal{N}},$ respectively, the \emph{squares} and the \emph{nonsquares} of $\F_{3^m}.$ For simplicity, $L'={\mathcal{Q}}\times \F_{3^m}\times \F_{3^m}.$
Thus $L'$ is a subgroup of ${\mathcal{R}}^*,$ of index $2.$
\subsection{Gray map}
The \emph{Gray} map $\phi$ from $R$ to $\F_3^{3}$ is defined by
$$\phi(a'+ub'+u^{2}c')=(a',b',c'),$$
for  $a',b',c' \in \F_3.$ It is a one to one map from $R$ to $\F_3^3.$
The \emph{Lee weight} of a vector $a+ub+u^{2}c$ is defined as the Hamming weight of its Gray image. That is to say,
$$w_L(\mathbf{a}+u\mathbf{b}+u^{2}\mathbf{c})=w_H(\mathbf{a})+w_H(\mathbf{b})+w_H(\mathbf{c}),$$
for  $\mathbf{a},\mathbf{b},\mathbf{c} \in \F_3^n.$
The \emph{Lee distance} of $x,y\in R^n$ is defined as $w_L(x-y).$
 So the Gray map is, by construction, a linear isometry from $(R^n,d_L)$ to $({\F^{3n}_3},d_H)$, where $d_{L},~d_{H}$ means Lee distance and Hamming distance, respectively.
For simplicity, we let throughout $N=3n.$
For future use, we note that scalars of weight one in $R$ comprize $\alpha u^{j},\alpha \in \F_{3}^*,j=0,1,2.$
\subsection{Codes}
A {\bf linear code} $C$ over $R$ of length $n$ is an $R$-submodule of $R^n$. If $x=(x_1,x_2,\cdots,x_n)$
 and $y=(y_1,y_2,\cdots,y_n)$ are two elements of  $R^n$, their Euclidean inner product
  is defined by $\langle x,y\rangle=\sum_{i=1}^nx_iy_i$, where the operation is performed in $R$. The {\bf dual code} of $C$ is denoted as $C^\perp=\{y\in R^n|\langle x,y\rangle =0, \forall x\in C\}.$ By definition, $C^\perp$ is also a linear code over $R$.
  Given a finite abelian group $G,$ a code over $R$ is said to be {\bf abelian} if it is an ideal of the group ring $R[G].$
  Namely, the coordinates of $C$ are indexed by elements of $G$ and $G$ acts regularly on this set.
  In the special case when $G$ is cyclic, the code is a cyclic code in the usual sense \cite{MS}.
  A code of length $\ell m,$ is said to be {\bf quasi-cyclic} of index $\ell$ and co-index $m,$ if it is linear and invariant under a shift of $\ell$ places.
  In particular, if the co-index is equal to three the code is said to be {\bf cubic} \cite{B}.
\section{Symmetry}
 For $a\in \mathcal{R}$, the vector $ev(a)$ is defined by the following evaluation map
$$ev(a)=(Tr(ax))_{x\in L},$$ where $L=L'$ or $L=\mathcal{R}^{*}$.

Define the code $C(m)$ by the formula $C(m)=\{ev(a)|a\in \mathcal{R}\}$. Thus $C(m)$ is a code of length $|L|$ over $R.$ Note that $|L'|=\frac{3^{3m}-3^{2m}}{2}$ and $|\mathcal{R}^{*}|={3^{3m}-3^{2m}}$.
\begin{lemma}\label{enum} If for all $x \in L$, we have $Tr(ax)=0$, then $a=0$.
\end{lemma}
\begin{proof}
When $L=L'$,
write $x=x_{1}+x_{2}(u-1)+x_{3}(u-1)^{2}$ and $a=a_{1}+a_{2}(u-1)+a_{3}(u-1)^{2}$ with $x_{1}\in \mathcal{Q},x_{2},x_{3},a_{1},a_{2},a_{3}$ in $\F_{3^m}.$ By a simple calculation we get
\begin{eqnarray*}
  ax &=& a_{1}x_{1}-a_{1}x_{2}+a_{1}x_{3}-a_{2}x_{1}+a_{2}x_{2}+a_{3}x_{1} \\
   && +(u-1)(a_{1}x_{2}+a_{1}x_{3}+a_{2}x_{1}+a_{2}x_{2}+a_{3}x_{1}) \\
   && +(u-1)^{2}(a_{1}x_{3}+a_{2}x_{2}+a_{3}x_{1})\\
   &=:& A_{1}+(u-1)A_{2}+(u-1)^{2}A_{2}
\end{eqnarray*}
 and $Tr(ax)=0$
is equivalent to $tr(A_{i})=0$, where $i=1,2,3$. The trace function $tr()$ \cite{MS} is nondegenerate, so we have $a_{j}=0$ where $j = 1,2,3$.
Consequently, $a=0.$
The case of $L=\mathcal{R}^{*}$ is similar to that of $L=L'$. Thus the proof is proved.
\end{proof}

\begin{proposition} The subgroup $L$ of the group of units $\mathcal{R}^*$ acts regularly on the coordinates of $C(m).$ \end{proposition}
\begin{proof}
For any $v,v' \in L,$ the change of variables $ x\mapsto (v'/v)x$ permutes the coordinates of $C(m),$ and maps $v$ to $v'.$
This defines a transitive action of $L$ on itself.
Such a permutation is unique, given $v,v'.$ This shows that the action is regular.
\end{proof}

The code $C(m)$ is thus an {\em abelian code} with respect to the group $L.$
In other words, it is an ideal of the group ring $R[L].$ As observed in the previous section that neither $\mathcal{R}^*,$ nor $L'$ are a cyclic group, hence $C(m)$ may be not cyclic.

It is immediate to see that, by the definition of Gray map, the ternary code $\phi(C(m))$ is quasi-cyclic of co-index $3.$

\section{Character sums}
This section is similar to that in \cite{SWLS}. For completeness, we collect it here. Let $\chi$ denote an arbitrary multiplicative character of $\F_q.$ Assume $q$ is odd.
Denote by $\eta$ the quadratic multiplicative character defined by $\eta(x)=1,$ if $x$ is a square and $\eta(x)=-1,$ if not.
Let $\psi$ denote the canonical additive character of $\F_q.$
The classical {\bf Gauss sum} can be defined by $$G(\chi)=\sum_{x\in \F_q^*}\psi(x)\chi(x).$$

Now, we give the following character sums
\begin{eqnarray*}
\Q=\sum_{x\in {\mathcal{Q}}}\psi(x), \ \N=\sum_{x\in {\mathcal{N}}}\psi(x).
\end{eqnarray*}
By orthogonality of characters \cite[Lemma 9, p. 143]{MS}, it is not difficult to check that $\Q+\N=-1.$
Since the characteristic function of ${\mathcal{Q}}$ is $\frac{1+\eta}{2},$ we obtain then
\begin{eqnarray*}
\Q=\frac{G(\eta)-1}{2}, \ \N=\frac{-G(\eta)-1}{2}.
\end{eqnarray*}

It is well known \cite{DY}, that if $q=3^m,$ the quadratic Gauss sum evaluates as

$$G(\eta)=(-1)^{m-1}i^m\sqrt{q}. $$

In particular, if $m$ is singly-even, these formulas can be simplified to
$G(\eta)=\epsilon(3)\sqrt{q},$ with $\epsilon(3)=(-1)^{\frac{(3+1)}{2}}=1,$
which implies
\begin{eqnarray*}
\Q=\frac{\sqrt{q}-1}{2},~ \N=-\frac{\sqrt{q}+1}{2}.
\end{eqnarray*}
In fact $\Q$ and $\N$ are examples of \emph{Gaussian periods}, and these relations could have been deduced from \cite[Lemma 11]{DY}.
\section{Weight distributions of trace codes}
Before we calculate the weight distribution of the trace code, let us first introduce a correlation lemma. Let $\omega=\exp(\frac{2\pi i}{3}).$ If $y=(y_1,y_2,\cdots,y_N)\in \mathbb{F}_3^N,$ let $$\Theta(y)=\sum_{j=1}^N\omega^{y_j}.$$
For convenience, we let $\theta(a)=\Theta(\phi(ev(a))).$ By linearity of the Gray map, and of the evaluation map, we see that $\theta(sa)=\Theta(\phi(ev(sa)))$ for any $s\in \F_3^*.$
\begin{lemma} [\cite{SWLS}, Lemma 1]\label{3.1} For all $y=(y_1,y_2,\cdots,y_N)\in \mathbb{F}_3^N,$ we have
$$\sum_{s=1}^{2}\Theta(sy)=2N-3w_H(y).$$
\end{lemma}

\subsection{The case when $L=L'$}

From Lemma \ref{3.1}, for $ev(a)\in C(m)$, by the definition of the Gray map, we have
\begin{equation}
  w_{L}(ev(a))=\frac{2N-\sum_{s=1}^{2}\Theta(\phi(ev(a)))}{3}=\frac{2N-\sum_{s=1}^{2}\theta(sa)}{3}.
\end{equation}
According to the value of $m$, we can obtain ternary codes with different weights. Now we present the weights of codewords of $C(m)$ by using Equation (1). The following theorem tells us that when $L=L'$ and $m$ is singly-even, $\phi(C(m))$ is a three-weight ternary linear code.
\subsubsection{$m$ is singly-even}
\begin{theorem}\label{enum1} Assume $m$ is singly-even. For $a\in \mathcal{R}$, the Lee weight of codewords of $C(m)$ is
given below:
\begin{enumerate}
\item[(a)] If $a=0$, then $w_L(ev(a))=0$;
\item[(b)] If $a=(u-1)^{2}a_3$, then if \\
~$a_{3} \in {\mathcal{Q}}$, then $w_L(ev(a))=3^{3m}-3^{5m/2}$,\\
~$a_{3} \in {\mathcal{N}}$, then $w_L(ev(a))=3^{3m}+3^{5m/2}$;
\item[(c)] If $a\in \mathcal{R}\backslash {\langle (u-1)^2\rangle}$, then $w_L(ev(a))=3^{3m}-3^{2m}$.
 \end{enumerate}

 \end{theorem}
\begin{proof}

Let $a=a_{1}+a_{2}(u-1)+a_{3}(u-1)^{2}$ with $a_{1},a_{2},a_{3} \in \mathbb{F}_{3^{m}}$,  $x=x_{1}+x_{2}(u-1)+x_{3}(u-1)^{2}$ with $x_{1} \in \mathcal{Q},~x_{2},x_{3} \in \mathbb{F}_{3^{m}}$, by a direct calculation we get
\begin{eqnarray*}
  ax &=& a_{1}x_{1}-a_{1}x_{2}+a_{1}x_{3}-a_{2}x_{1}+a_{2}x_{2}+a_{3}x_{1}\\
    &&+(a_{1}x_{2}+a_{1}x_{3}+a_{2}x_{1}+a_{2}x_{2}+a_{3}x_{1})u +(a_{1}x_{3}+a_{2}x_{2}+a_{3}x_{1})u^{2}\\
   &=:&b_{1}+b_{2}u+b_{3}u^{2}.
\end{eqnarray*}
Thus we have
\begin{eqnarray*}
\phi(ev(a))&=&(tr(b_{1}),tr(b_{2}),tr(b_{3}))_{x_{1},x_{2},x_{3}}, \\
\theta(a)&=&\sum_{x_{1},x_{2},x_{3} }\omega^{tr(b_{1})}+
    \sum_{x_{1},x_{2},x_{3}}\omega^{tr(b_{2})}+\sum_{x_{1},x_{2},x_{3}}\omega^{tr(b_{3})}.
\end{eqnarray*}
     Since $m$ is singly-even, $s\in \F_3^*$ is a square in $\F_{3^m},$ we have $\theta(sa)=\theta(a)$ for any $s\in \F_3^*.$
\begin{enumerate}
\item[(a)] If $a=0$, then $Tr(ax)=0$. So $w_L(ev(a))=0$.
\item[(b)] When $a=(u-1)^{2}a_{3}$ with $a_{3} \in {\mathcal{Q}},$ then $\theta(a)=3^{2m+1}\overline{Q}$. Thus $w_L(ev(a))=3^{3m}-3^{5m/2}$ by Equation $(1)$.\\
    When $a=(u-1)^{2}a_{3}$ with $a_{3} \in {\mathcal{N}}, \theta(a)=3^{2m+1}\overline{N}$.
    Then we deduce from the Equation $(1)$ that $w_L(ev(a))=3^{3m}+3^{5m/2}$.
\item[(c)] When $a\in \mathcal{R}\backslash {\langle (u-1)^2\rangle}, \theta(a)=0$. Then we have $w_L(ev(a))=3^{3m}-3^{2m}$
\end{enumerate}
\hspace*{1.3cm}by Equation $(1)$.
\end{proof}
\hspace{-0.4cm}\textbf{Remark 5.1} By the point of Lemma \ref{enum} and Theorem \ref{enum1}, a family of ternary linear code of length $N=(3^{3m+1}-3^{2m+1})/2,$ dimension $3m,$ with three nonzero weights $w_1<w_2<w_3$ of values has been constructed. In detail, we list the weight distribution of $\phi(C_{m})$ in Table I. Notice that the parameters are different from those in \cite{DGZ,DLLZ,SWLS}. Thus, the obtained code in Theorem 5.2 are new.
\\

\begin{center}$\mathrm{Table~ I. }~~~\mathrm{weight~ distribution~ of}~ \phi(C_{m}) $\\
\begin{tabular}{cccc||cc}
\hline
  Weight&&   & & Frequency  \\
  \hline

  0        & &   & & 1\\
  $\ \ w_1=3^{3m}-3^{5m/2}$        & &   &              &$f_1=\frac{3^m-1}{2}$\\
  $w_2=3^{3m}-3^{2m}$  &    & &       & $\ \ \ \ \ f_2=3^{3m}-3^{m}$ \\
  $\ \  w_3=3^{3m}+3^{5m/2}$  &    & &       &$f_3=\frac{3^m-1}{2}$ \\
  \hline
\end{tabular}
\end{center}

\hspace{-0.4cm}\text{\textbf{Example 5.1}} Let $m=2$. We obtain a ternary code of parameters $[972,6,486]$. The nonzero weights are $486,~648$ and $972$, of frequencies $4,~720$ and $4$, respectively.

\subsubsection{$m$ is odd }

Note that $G(\eta)$ is imaginary, which implies that $\Re(\Q)=\Re(\N)=-\frac{1}{2}.$ The following lemma is a special case of Lemma 2 in \cite{SWLS}.
\begin{lemma}[\cite{SWLS}, Lemma 2]
 \label{h} $\sum_{s=1}^{2}\theta(sa)=2\Re(\theta(a)).$
\end{lemma}

 By a similar approach in the proof of Theorem \ref{enum1} and combining Lemma \ref{h}, it is not difficult to obtain the following theorem.

\begin{theorem}\label{enum2} Assume $m$ is odd. For $a\in \mathcal{R}$, the Lee weight of codewords of $C(m)$ is
given below:
\begin{enumerate}
\item[(a)] If $a=0$, then $w_L(ev(a))=0$;
\item[(b)] If $a=(u-1)^{2}a_3$, with $a_{3}\in \F_{3^{m}}^{*}$,then $w_L(ev(a))=3^{3m}$;
\item[(c)] If $a\in \mathcal{R}\backslash {\langle (u-1)^2\rangle}$, then $w_L(ev(a))=3^{3m}-3^{2m}$.
 \end{enumerate}
 \end{theorem}
\begin{proof}
 The proofs of the cases (a) and (c) are similar to that of Theorem \ref{enum1}. The case (b) follows from Lemma \ref{h} applied to the correlation lemma.
Thus  $\Re(\theta(a))=-3^{2m+1}/2,$ and $3w_L(ev(a))=2N-2\Re(\theta(a)),$ which implies $w_L(ev(a))=3^{3m}.$
The result follows.
\end{proof}

\hspace{-0.4cm}\text{\textbf{Example 5.2}} Let $m=1$. We obtain a ternary code of parameters $[27,3,18]$. The nonzero weights are $18$ and $27$, of frequencies $24$ and $2$, respectively.

\subsection{The case when $L=\mathcal{R}^{*}$}
In this subsection, we will construct an infinite family of two-weight codes on the condition that $L=\mathcal{R}^{*}$.
\begin{theorem}\label{enum2} For $a\in \mathcal{R}$, the Lee weight of codewords of $C(m)$ is given below:
\begin{enumerate}
\item[(a)] If $a=0$, then $w_L(ev(a))=0$;
\item[(b)] If $a=(u-1)^{2}a_3$, with $a_{3}\in \F_{3^{m}}^{*}$, then $w_L(ev(a))=2\cdot3^{3m}$;
\item[(c)] If $a\in \mathcal{R}\backslash {\langle (u-1)^2\rangle}$, then $w_L(ev(a))=2(3^{3m}-3^{2m})$.
 \end{enumerate}
 \end{theorem}
 \begin{proof}
 By a similar approach in the proof of Theorem \ref{enum1}, the result follows, so we omit the details here.
 \end{proof}

\hspace{-0.4cm}\text{\textbf{Example 5.3}} Let $m=1$. We obtain a ternary code of parameters $[54,3,36]$. The nonzero weights are $36$ and $54$, of frequencies $24$ and $2$, respectively.\\

\hspace{-0.4cm}\text{\textbf{Example 5.4}} Let $m=2$. We obtain a ternary code of parameters $[1944,6,1296]$. The  nonzero weights are $1296$ and $1458$, of frequencies $720$ and $8$, respectively.\\

\hspace{-0.4cm}\textbf{Remark 5.2} Comparing with \cite{CK,SWS,SWLS}, we obtain two new families of ternary two-weight codes based on the obtained trace codes over $R$ in Theorem 5.4 and Theorem 5.5. The weight distribution is listed in Table II.\\

 \begin{center}$\mathrm{Table \  II. }~~~\mathrm{weight~ distribution~ of}~ \phi(C_{m})$ {\rm from Theorems 5.4 and 5.5}\\
\begin{tabular}{ccc||cc}
\hline
  Weight& &&& Frequency  \\
  \hline
  0        & &&& 1\\
  $\ \ \ \ \ \ \ \ \ w_1'=3^{3m}-3^{2m}$        &  &&&              $\ \ \  f_1'=3^{3m}-3^{m}$\\
  $w_2'=3^{3m}$            &  &&&              $f_2'=3^m-1$\\
  \hline
   \hline
  0        & &&& 1\\
  $\ \ \ \ \ \ \ \ \ w_1''=2(3^{3m}-3^{2m})$        &  &&&              $\ \ \  f_1''=3^{3m}-3^{m}$\\
  $w_2''=2\cdot3^{3m}$            &  &&&              $f_2''=3^m-1$\\
  \hline
\end{tabular}
\end{center}

\section{Optimality of the image codes}
In the previous section, we have constructed two new infinite families of ternary two-weight codes and a new family of three-weight ternary codes. Now we study their optimality.
\begin{theorem}\label{t}
The image codes $\phi(C(m))$ of length $3|L|$ are optimal based on the following cases
\begin{enumerate}
\item[(i)] $m$ is odd and $L=L'$ in Theorem $5.4$;
\item[(ii)] $m$ is a positive integer and $L=\mathcal{R}^{*}$ in Theorem $5.5$.
\end{enumerate}
\end{theorem}
\begin{proof}
Recall the $3$-ary version of the Griesmer bound. If $[N,K,d]$ are the parameters of a linear ternary code. Then
$$\sum_{j=0}^{K-1}\Big{\lceil} \frac{d}{3^j} \Big{\rceil} \le N.$$
In the case of (i), $N=(3^{3m+1}-3^{2m+1})/2,\, K=3m,\,d=3^{3m}-3^{2m}.$ The ceiling function takes two values depending on the position of $j$.
\begin{itemize}
 \item $0\leq j\le 2m \Rightarrow \lceil \frac{d+1}{3^j} \rceil =3^{3m-j}-3^{2m-j}+1$.
 \item $2m<j\leq3m-1 \Rightarrow \lceil \frac{d+1}{3^j} \rceil =3^{3m-j}$.
\end{itemize}
\begin{eqnarray*}
\sum_{j=0}^{K-1}\Big{\lceil} \frac{d+1}{3^j} \Big{\rceil} &=& \sum_{j=0}^{2m}(3^{3m-j}-3^{2m-j}+1) + \sum_{j=2m+1}^{3m-1}3^{3m-j} \\
&=&(3^{3m+1}-3^{2m+1})/2+2m-1 > N,
\end{eqnarray*}
 which collapses to $m \geq 1$. Thus, this completed the proof of the case (i).

For case (ii), we have  $N=(3^{3m+1}-3^{2m+1}),\, K=3m,\,d=2(3^{3m}-3^{2m}).$  By a simple calculation, we can easily obtain that
\begin{eqnarray*}
\sum_{j=0}^{K-1}\Big{\lceil} \frac{d+1}{3^j} \Big{\rceil} &=& \sum_{j=0}^{2m}(2(3^{3m-j}-3^{2m-j})+1) + \sum_{j=2m+1}^{3m-1}2\cdot3^{3m-j} \\ &=&3^{3m+1}-3^{2m+1}+2m-1 > N,
\end{eqnarray*}
 which collapses to $m \geq 1$.
\end{proof}
\section{ The dual Lee distance of the trace code}
Similar to Lemma 3 and Theorem 4 in \cite{SWLS}, we can calculate the dual distance of the obtained trace code. We still prove them for completeness.
\begin{lemma}\label{nonde}
 If for all $a \in \mathcal{R},$ we have that $Tr(ax)=0,$ then $x=0.$
\end{lemma}

\begin{proof}
The proof is similar to that of Lemma \ref{enum}, we omit it here.
\end{proof}
\begin{theorem} \label{d'}
 For all $m\ge 1,$ the dual Lee distance $d'$ of $C(m)$ is $2.$
\end{theorem}

\begin{proof}
 First, we check that $d'\ge 2.$
 The approach is the same as Theorem 7.2 in \cite{SWLS}. Now we will prove it by showing that $C(m)^\bot$ does not contain a codeword of Lee weight one.
 If it does, let us assume first that it has value $\alpha u^{j}\neq 0$ at some $x \in L$, where $j={0,1,2}$. This implies that $\forall a \in \mathcal{R},\alpha u^{j}
 Tr(ax)=0,$ or, $Tr(a \alpha u^{j}x)=0,$ and by Lemma \ref{nonde}, we have $ x=0$, which contradicts the assumption. So $d'\ge 2$.\\
\hspace*{0.5cm}~Next, we show that $d'<3.$ If not, we can apply the sphere-packing bound to $\phi(C(m)^\bot).$ Naturally, we obtain
 $$3^{3m}\ge 1+|L|(3-1),$$  where $|L|=|L'|$ or $|L|=|\mathcal{R}^{*}|$. That implies $(3-2\cdot 3^{m})3^{2m} \geq 1$ or $3^{2m+1}(2-5\cdot3^{m-1})\geq 1$, respectively. It is a contradiction when $m\geq 1$. In summary, $d=2$.
 \end{proof}
\section{Application to secret sharing schemes}
\subsection{Determining minimal vectors}
It is interesting to determine minimal vectors of a given $p$-ary linear code. Minimal vectors in linear codes arise in numerous applications, particularly, in studying linear secret sharing schemes (SSS). We say that a minimal vector of a linear code $C$ is a nonzero codeword that does not cover any other nonzero codeword. The basic property minimal vector of a given $p$-ary linear code was described by the following lemma \cite{AB}.

\begin{lemma}\label{d}(Ashikhmin-Barg) Denote by $w_0$ and $w_{\infty}$ the minimum and maximum nonzero weights of a $p$-ary linear code $C$, respectively. If
$$\frac{w_0}{w_{\infty}}>\frac{p-1}{p},$$ then every nonzero codeword of $C$ is minimal.
\end{lemma}

Next, under Lemma \ref{d}, we investigate the minimal vectors of the three-weight and two-weight codes which constructed in Section $5$.
\begin{proposition}\label{support1}
 All the nonzero codewords of the image codes $\phi(C(m))$ are minimal based on the following three cases
\begin{enumerate}
\item[(i)] $m(\geq 6)$ is singly-even and $L=L'$ in Theorem $5.2$.\\
\item[(ii)] $m(\geq 1)$ is odd and $L=L'$ in Theorem $5.4$.\\
\item[(iii)] $m$ is a positive integer and $L=\mathcal{R}^{*}$ in Theorem $5.5$.
\end{enumerate}
\end{proposition}
\begin{proof}
 In the case of (1), by the preceding lemma with $w_0=w_1,$ and $w_{\infty}=w_3.$ Rewriting the inequality of the lemma as $3w_1>(3-1)w_3,$ we end up with the condition
 $$ 3\cdot3^{3m}-3\cdot3^{5m/2}>2\cdot3^{3m}+2\cdot3^{5m/2}.$$ Since $m$ is singly-even. The condition follows from the fact that $3^{m/2}\geq5$.

The proof of other cases are similar to those of the case (1), we omit them here.
\end{proof}
\subsection{Secret sharing schemes}
The purpose of determining minimal vectors is to determine the sets of all minimal access sets of a secret sharing scheme (SSS). An SSS were first introduced by Blakly and Shamir at the end of $70s$ twentieth antury. Massey's scheme is one of the famous SSS. Massey's scheme is a construction of such a scheme where a code $C$ of length $N$ over $\mathbb{F}_{p}$. On the other hand, it is worth mentioning when all nonzero codewords are minimal, it was shown in \cite{DY2} that there is the following alternative, depending on $d'$:
\begin{itemize}
 \item If $d'\ge 3,$ then the SSS is \emph{``democratic''}: every user belongs to the same number of coalitions.
 \item If $d'=2,$  then there are users  who belong to every coalition: the \emph{``dictators''.}
\end{itemize}
Depending on the application, one or the other situation might be more suitable.
By Proposition \ref{support1} and Theorem \ref{d'}, we see that three Secret Sharing Schemes built on the image codes $\phi(C(m))$ in Theorems 5.2, 5.4 and 5.5 are dictatorial.
\section{Conclusion}
This paper is devoted to the study of trace codes with defining set $L$ included in an extension of degree $m$ of the alphabet, the local ring $\mathbb{F}_3+u\mathbb{F}_3+u^{2}\mathbb{F}_{3}$ with $u^{3}=1$.  These codes are abelian, and their ternary images are quasi-cyclic of co-index three (a.k.a. cubic codes).
Their Lee weight distributions are computed by using Gauss sums (see Table I and Table II).
 These codes have three nonzero weights when $m$ is singly-even and $|L|=\frac{3^{3m}-3^{2m}}{2}.$
 When $m$ is odd, and $|L|=\frac{3^{3m}-3^{2m}}{2}$, or $|L|={3^{3m}-3^{2m}}$ and $m$ is a positive integer,
 we obtain two new infinite family of two-weight codes. Both are shown to be optimal, by application of the Griesmer bound.
 Applications of the image codes to secret sharing schemes are also given.

It would be interesting to replace our Gaussian periods $\Q,\N$  by other character sums that are amenable to exact evaluation, in the vein of the sums which appear in the study of irreducible
cyclic codes \cite{DY,MR}. This would lead to other enumerative results of codes with few weights. The parameters of the two-weight codes we constructed are different from
those in the classic paper \cite{CK}, and also from more recent constructions like \cite{DD,SLS2, SLS1}. Writing a new survey on two-weight codes seems like a challenging
but very useful project.
\section{Acknowledgement}
 This research is
supported by National Natural Science Foundation of China (61672036), the Open Research Fund of National Mobile Communications Research Laboratory, Southeast University (2015D11),
Technology Foundation for Selected Overseas Chinese Scholar, Ministry of Personnel of China (05015133) and
Key Projects of Support Program for outstanding young talents in Colleges and Universities (gxyqZD2016008).

\end{document}